\documentclass[11pt,reqno]{amsart}
\usepackage[T1]{fontenc}
\usepackage[utf8]{inputenc}
\usepackage{lmodern}
\usepackage{amsmath,amssymb,amsthm,mathtools}
\usepackage[margin=1.1in,includefoot,headheight=0pt,headsep=0pt,footskip=18pt]{geometry}
\usepackage{booktabs}
\usepackage{placeins}
\usepackage[colorlinks=true,linkcolor=blue,citecolor=blue,urlcolor=blue]{hyperref}
\makeatletter
\def\ps@plain{\ps@empty
  \def\@oddfoot{\normalfont\small\hfil\thepage\hfil}
  \let\@evenfoot\@oddfoot}
\let\ps@firstpage\ps@plain
\makeatother
\AtBeginDocument{\pagestyle{plain}}
\newtheorem{theorem}{Theorem}[section]
\newtheorem{lemma}[theorem]{Lemma}
\newtheorem{proposition}[theorem]{Proposition}
\newtheorem{corollary}[theorem]{Corollary}
\theoremstyle{remark}

\newcommand{\F}{\mathbb F}
\newcommand{\bch}{\operatorname{BCH}}

\DeclareMathOperator{\rank}{rank}
\DeclareMathOperator{\spanof}{span}
\title[Generalized packing and covering]{Auxiliary Codes and the Generalized Packing--Covering Conjecture}
\author{Isaac Barouch Essayag}
\author{Aryeh Lev Zabokritskiy (Yohananov)}
\date{}
\hypersetup{pdftitle={Auxiliary Codes and the Generalized Packing-Covering Conjecture},pdfauthor={Isaac Barouch Essayag; Aryeh Lev Zabokritskiy (Yohananov)},pdflang={en}}
\begin{document}
\begin{abstract}
The generalized packing--covering conjecture asks whether, at every order, the packing radius of a linear code is at most its covering radius. We prove the conjecture for every linear code of redundancy at most fourteen over every finite field, extending the previously established range of redundancies at most seven. We also prove the bound on generalized Hamming weights $d_t(C)\le2R_t(C)+1$ whenever the alphabet size $q$ satisfies $q\ge R_t(C)$. Both results use an auxiliary-code criterion that converts a covering property in the syndrome space into a weight bound. For binary primitive BCH codes, the packing radius is strictly smaller than the covering radius for every fixed error parameter and order, both at least two, once the extension degree is sufficiently large; this follows from existing covering bounds.
\end{abstract}
\maketitle

\section{Introduction}

The classical packing radius of a code does not exceed its covering radius. Elimelech, Firer, and Schwartz introduced generalized covering radii and conjectured an analogous comparison at every order~\cite{EFS2021}. The conjecture relates two hierarchies: generalized Hamming weights, which measure the smallest support of a subcode of prescribed dimension, and generalized covering radii, which measure the support needed to cover several ambient words simultaneously. Subsequent work bounded these radii for Reed--Muller codes~\cite{EWS2022} and determined the optimal asymptotic rate of second-order covering codes without assuming linearity~\cite{ES2024}. Li, Shangguan, and Wei determined that rate at every fixed order, both with and without the linearity requirement~\cite{LSW2026}.

Let $C$ be an $[n,k]_q$ linear code, and let $1\le t\le\min\{k,n-k\}$. Its $t$th generalized Hamming weight $d_t(C)$ is the least number of coordinates supporting a $t$-dimensional subcode. Its $t$th generalized covering radius $R_t(C)$ is the least integer $r$ such that every $t$ ambient words can be moved into $C$ by changing entries in a common set of at most $r$ coordinates. With
\[
\delta_t(C)=\left\lfloor\frac{d_t(C)-1}{2}\right\rfloor,
\]
the conjecture is
\begin{equation}\label{eq:conjecture}
\delta_t(C)\le R_t(C),
\qquad\text{or equivalently}\qquad
d_t(C)\le 2R_t(C)+2.
\end{equation}
The first-order case is the classical comparison. The higher-order formulation is not an immediate consequence of packing disjoint Hamming balls; the rank restriction inherent in generalized weights changes that argument~\cite[Section VI]{EFS2021}.

The comparison also transfers information between the two hierarchies. As emphasized by Yu and Schwartz~\cite[Introduction]{YS2026}, generalized Hamming weights have a much more developed theory than generalized covering radii. Whenever~\eqref{eq:conjecture} holds, lower bounds on generalized Hamming weights yield lower bounds on the corresponding covering radii. Proving the conjecture for a class of codes therefore makes existing weight results available for its covering problem, in addition to extending the classical geometric relation between packing and covering.

Yu and Schwartz~\cite{YS2026} proved the conjecture at order two, for codes of rate at most $3/5$, and asymptotically for every fixed rate ceiling below one. They also proved it for all orders when the redundancy $n-k$ is at most seven.

\noindent\textbf{Our contributions.}
\begin{enumerate}
\item We prove the conjecture for every linear code of redundancy at most fourteen over every finite field, extending the previous uniform range of seven.
\item We prove $d_t(C)\le2R_t(C)+1$ whenever $q\ge R_t(C)$, as a consequence of a general auxiliary-code criterion.
\item For binary primitive BCH codes with fixed error parameter $e\ge2$ and order $t\ge2$, we obtain a packing radius strictly smaller than the covering radius for all sufficiently large extension degrees, using existing covering bounds for adjacent BCH families.
\end{enumerate}

The first contribution has the following uniform form.

\begin{theorem}\label{thm:redundancy}
Let $C$ be an $[n,k]_q$ linear code with $n-k\le14$. Then
\[
d_t(C)\le2R_t(C)+2
\qquad\text{for every }1\le t\le\min\{k,n-k\}.
\]
\end{theorem}

The theorem imposes no restriction on the length or on the finite field. Its proof combines the established cases of Yu and Schwartz with an auxiliary-code construction and classical generalized-weight estimates. The auxiliary construction applies beyond bounded redundancy: a code of suitable length, dimension, and minimum distance supplies a syndrome subspace whose short cover forces a subcode of $C$ with small support.

One consequence is a large-alphabet bound with a smaller additive constant.

\begin{theorem}\label{thm:alphabet}
Let $C$ be an $[n,k]_q$ linear code, let $1\le t\le\min\{k,n-k\}$, and put $r=R_t(C)$. If $q\ge r$, then
\[
d_t(C)\le2r+1.
\]
\end{theorem}

The geometric description of generalized covering by spans of parity-check columns is central to this argument; related geometric formulations were developed by Alfarano, Marino, Neri, and Trombetti~\cite{AMNT2026}. The BCH application concerns fixed-error families whose rates tend to one, outside the fixed-rate regime above.

Section~2 sets out the column formulation and the prior results used in the proof. Section~3 develops the auxiliary-code criterion. Section~4 proves the small-redundancy theorem by bounding the dimension of a hypothetical counterexample. Section~5 treats the BCH consequence, and Section~6 discusses the remaining cases.

\section{Syndrome spaces and established bounds}\label{sec:prelim}

We use parity-check columns to express both sides of~\eqref{eq:conjecture}. This lets us construct a syndrome space for the covering problem and then measure the dimension of the subcode forced by its cover.

Write $\rho=n-k$ for the redundancy, and let
\[
H=(h_1,\ldots,h_n)\in\F_q^{\rho\times n}
\]
be a full-row-rank parity-check matrix of $C$. For a coordinate set $S\subseteq\{1,\ldots,n\}$, write $H_S$ for the matrix of columns indexed by $S$ and $V_S=\spanof_{\F_q}\{h_i:i\in S\}$. For a vector, its support is the set of its nonzero coordinates; for a subspace, its support is the union of the supports of its vectors. The subcode of $C$ supported on $S$ is the kernel of $H_S$, with zero coordinates inserted outside $S$. Consequently,
\begin{equation}\label{eq:weight-nullity}
d_t(C)=\min\{|S|:|S|-\rank(H_S)\ge t\}.
\end{equation}
This is the standard column description of generalized Hamming weights~\cite{Wei1991}.

The covering definition has the corresponding syndrome form~\cite[Section III, Definition 1 and Lemma 5]{EFS2021}:
\begin{equation}\label{eq:cover-span}
R_t(C)=\max_{\substack{U\le\F_q^\rho\\\dim U=t}}
\min\{|S|:U\subseteq V_S\}.
\end{equation}
Indeed, correcting entries in $S$ means expressing the prescribed syndromes using columns indexed by $S$. Tuples whose syndrome span has dimension less than $t$ are covered by extending that span to a $t$-space. In particular,
\begin{equation}\label{eq:basic-bounds}
t\le R_t(C)\le\rho,
\qquad d_t(C)\le\rho+t,
\end{equation}
where the last inequality is the generalized Singleton bound. Zero columns and repeated columns are allowed throughout.

We will also use the ball-covering inequality. For integers $Q\ge2$ and $0\le r\le n$, define
\[
V_Q(n,r)=\sum_{i=0}^r\binom ni(Q-1)^i.
\]
Counting the $t\times n$ error matrices supported on at most $r$ columns gives
\begin{equation}\label{eq:sphere}
V_{q^t}(n,R_t(C))\ge q^{t\rho}.
\end{equation}
Here the right side is the number of $t$-tuples of syndromes. We use the convenient estimate
\begin{equation}\label{eq:binomial-volume}
V_Q(n,r)\le\binom nr Q^r.
\end{equation}
To verify it, expand $Q^r=(1+(Q-1))^r$ and use
$\binom ni\le\binom nr\binom ri$ for $0\le i\le r$.

For clarity, we collect the exact prior cases needed later. Each item below gives~\eqref{eq:conjecture} for the indicated range:
\begin{enumerate}
\item $t=1$, by the classical packing--covering comparison;
\item $t=2$, by~\cite[Theorem 5]{YS2026};
\item $R_t(C)=t$, by~\cite[Lemma 6]{YS2026};
\item $t\ge\rho-4$, by~\cite[Lemma 7]{YS2026};
\item $t\ge3$ and $k\le5t-2$, by~\cite[Lemma 9]{YS2026};
\item $3R_t(C)\ge2\rho$, by~\cite[Lemma 11]{YS2026}.
\end{enumerate}
The first four items imply the redundancy-seven case~\cite[Corollary 8]{YS2026}.

The proof of Theorem~\ref{thm:redundancy} has two parts. First, we choose a syndrome space using a short auxiliary code; its minimum distance ensures that a short cover creates enough column nullity to apply~\eqref{eq:weight-nullity}. Second, for the remaining parameters, a hypothetical violation of~\eqref{eq:conjecture} forces a large generalized weight in the dual code. A classical weight estimate then bounds the dimension, and~\eqref{eq:sphere} excludes the remaining lengths.

\section{The auxiliary-code criterion}\label{sec:auxiliary}

To obtain an upper bound on $d_t(C)$ from covering, it is enough to choose one syndrome space whose short cover produces a set of columns with nullity at least $t$. The next result specifies how an auxiliary code makes this possible.

\begin{theorem}\label{thm:auxiliary}
Let $C$ be an $[n,k]_q$ linear code, let $1\le t\le\min\{k,\rho\}$ with $\rho=n-k$, and put $r=R_t(C)$. Suppose there is an $[L,t,d]_q$ linear code $D$ such that
\[
r<L\le\rho,
\qquad d\ge r-t+2.
\]
Then $d_t(C)\le r+L$.
\end{theorem}

For the proof, choose a set of coordinate indices $A=\{a_1,\ldots,a_L\}\subseteq\{1,\ldots,n\}$ whose columns in $H$ are independent. The map
\[
\phi:\F_q^L\longrightarrow V_A,\qquad
\phi(x_1,\ldots,x_L)=\sum_{\ell=1}^L x_\ell h_{a_\ell}
\]
is an isomorphism. We use $U=\phi(D)$ as the syndrome space to be covered. Its cover and the original columns will produce the coordinate set required by~\eqref{eq:weight-nullity}.

\begin{proof}
By~\eqref{eq:cover-span}, there is a set of coordinate indices $T\subseteq\{1,\ldots,n\}$ with $|T|\le r$ and $U\subseteq V_T$. Removing indices of dependent columns preserves its span, so we may assume that the columns indexed by $T$ are independent.

Put $J=A\cap T$, $j=|J|$, and $b=\dim(U\cap V_J)$. If $b>0$, the inverse image of $U\cap V_J$ is a $b$-dimensional subcode of $D$ supported on $j$ coordinates. Its restriction to those coordinates has minimum distance at least $d$, so Singleton gives $j-b+1\ge d$. It follows that
\[
\dim(U+V_J)=t+j-b\ge t+d-1\ge r+1.
\]
But $U+V_J\subseteq V_T$ and $\dim V_T=|T|\le r$, a contradiction. Thus $b=0$.

Both $U$ and $V_J$ lie in $V_A\cap V_T$, and their intersection is zero. Hence $\dim(V_A\cap V_T)\ge t+j$. The independence of $A$ and $T$ now gives
\begin{align*}
|A\cup T|-\rank(H_{A\cup T})
&=L+|T|-j-\bigl(L+|T|-\dim(V_A\cap V_T)\bigr)\\
&=\dim(V_A\cap V_T)-j\ge t.
\end{align*}
Applying~\eqref{eq:weight-nullity} yields
$d_t(C)\le|A\cup T|\le L+r$.
\end{proof}

The length of the auxiliary code determines the resulting bound. An auxiliary of length $r+1$ gives $2r+1$, whereas length $r+2$ gives the conjectured $2r+2$. The distance requirement is the same in both cases.

\begin{proof}[Proof of Theorem~\ref{thm:alphabet}]
If $r=\rho$, Singleton gives $d_t(C)\le\rho+t\le2r$. Otherwise $r+1\le\rho$. Since $q\ge r$, an extended Reed--Solomon code provides an $[r+1,t,r-t+2]_q$ auxiliary code~\cite{MS1977}. Explicitly, evaluate polynomials of degree at most $t-1$ at $r$ distinct elements of $\F_q$, and append their coefficient of $x^{t-1}$. A nonzero polynomial gives at most $t-1$ zero coordinates, including the appended coordinate, so the minimum distance is at least $r-t+2$. Theorem~\ref{thm:auxiliary} gives $d_t(C)\le2r+1$.
\end{proof}

Small auxiliaries also work over fields that do not satisfy the large-alphabet condition. The following four pairs are the ones needed in the bounded-redundancy proof.

\begin{corollary}\label{cor:four-pairs}
For every finite field,~\eqref{eq:conjecture} holds whenever
\[
(t,R_t(C))\in\{(3,4),(3,5),(4,5),(4,6)\}.
\]
\end{corollary}

\begin{proof}
Put $r=R_t(C)$. If $\rho\le r+1$, Singleton gives $d_t(C)\le\rho+t\le2r+1$. For $\rho\ge r+2$, Theorem~\ref{thm:auxiliary} reduces the four cases to the existence of codes with parameters
\[
[6,3,\ge3]_q,\quad[7,3,\ge4]_q,\quad
[7,4,\ge3]_q,\quad[8,4,\ge4]_q,
\]
respectively.

For the first and third codes, take a $3\times L$ matrix whose columns represent $L=6$ or $7$ distinct one-dimensional subspaces of $\F_q^3$, including a basis. There are $q^2+q+1\ge7$ such subspaces. Its kernel has dimension $L-3$ and minimum distance at least three, since no column is zero and no two are proportional.

For the $[7,3,\ge4]_q$ code, take the seven nonzero vectors in $\{0,1\}^3$ as generator columns. A nonzero linear functional vanishes on at most four points of the full cube: pairing along a coordinate with a nonzero coefficient gives at most one zero per pair. Since the origin is a zero, at most three of the seven columns are zeros. The generator has rank three, and every nonzero word has weight at least four.

Finally, evaluate affine linear functions on the eight points of $\{0,1\}^3$. The evaluation map has rank four, as can be checked at the origin and the three standard basis vectors. Every nonconstant affine function has at most four zeros by the same pairing argument, and a nonzero constant has no zeros. This gives the last auxiliary code.
\end{proof}

For example, if $R_3(C)=4$, the six-column auxiliary above forces $d_3(C)\le10$ whenever $\rho\ge6$; Singleton handles smaller redundancy. The same construction is available for every alphabet and every length of $C$. The method requires an auxiliary with the stated parameters. A binary $[8,5,3]$ code does not exist: a rank-three parity-check matrix would need eight distinct nonzero columns, whereas $\F_2^3$ has only seven. Thus auxiliary existence itself imposes a genuine restriction on this sufficient condition.

\section{Codes of redundancy at most fourteen}\label{sec:small}

We now prove Theorem~\ref{thm:redundancy}. The cases supplied by Section~\ref{sec:prelim} and Corollary~\ref{cor:four-pairs} leave finitely many pairs of order and covering radius for each redundancy. To deal with arbitrarily long codes, we first show that a failure of the conjecture bounds the dimension of the code.

Recall that $\rho=n-k$ and $r=R_t(C)$. A failure means $d_t(C)\ge2r+3$. Duality for generalized weights translates this into a lower bound on a weight of $C^\perp$~\cite{Wei1991}. Combining that lower bound with the classical shortening estimate of Tsfasman and Vl\u{a}du\c{t}~\cite{TV1995}, in the form recalled in~\cite[Lemma 2]{YS2026}, yields the following explicit form. We include the short derivation to identify the indices and the uniform length bound used here.

\begin{lemma}\label{lem:dimension-cap}
Suppose $d_t(C)\ge2r+3$, where $r=R_t(C)$, and set
\[
a=\rho-2r+t-2,
\qquad b=\rho-2r-1.
\]
Then $1\le a<\rho$, and
\begin{equation}\label{eq:dimension-cap}
k\le K(q,\rho,t,r):=
\min_{a<h\le\rho}
\left\lfloor
\frac{(q^h-q^{h-a})h-b(q^h-1)}{q^{h-a}-1}
\right\rfloor.
\end{equation}
The minimum is over integers $h$.
\end{lemma}

\begin{proof}
Singleton gives $2r+3\le\rho+t$, hence $a\ge1$. Since $r\ge t$, we also have $a<\rho$.

We first obtain the needed dual-weight inequality directly. An $a$-dimensional subcode of $C^\perp$ is generated by $BH$ for a rank-$a$ matrix $B\in\F_q^{a\times\rho}$. If its support has size at most $n-2r-2$, at least $2r+2$ columns of $H$ lie in $\ker B$, whose dimension is $\rho-a=2r-t+2$. Those columns have nullity at least $t$, contradicting $d_t(C)\ge2r+3$. Therefore
\begin{equation}\label{eq:dual-lower}
d_a(C^\perp)\ge n-2r-1=k+b.
\end{equation}

Fix an integer $h$ with $a<h\le\rho$. Shorten $C^\perp$ at $\rho-h$ independent coordinates of its generator matrix $H$. The resulting code has dimension $h$ and length at most $k+h$. For a uniformly chosen $a$-dimensional subcode of this shortened code, a fixed active coordinate belongs to its support with probability
\[
\lambda=\frac{q^h-q^{h-a}}{q^h-1}.
\]
Indeed, the proportion of $a$-subspaces contained in the kernel of a nonzero coordinate functional is $(q^{h-a}-1)/(q^h-1)$, by counting incidences between $a$-subspaces and hyperplanes. Some $a$-subcode consequently has support at most $\lambda(k+h)$. With~\eqref{eq:dual-lower}, this gives
\[
k+b\le\frac{q^h-q^{h-a}}{q^h-1}(k+h).
\]
Rearranging and taking the integer part proves the bound for each $h$, and hence their minimum.
\end{proof}

The key point is that~\eqref{eq:dimension-cap} bounds every possible dimension under the failure assumption. With $N=\rho+K(q,\rho,t,r)$, we have $n\le N$, so the inequality
\begin{equation}\label{eq:finite-exclusion}
\binom Nr<q^{t(\rho-r)}
\end{equation}
contradicts~\eqref{eq:sphere} by~\eqref{eq:binomial-volume}.

\begin{proof}[Proof of Theorem~\ref{thm:redundancy}]
Assume a counterexample with $\rho\le14$, and put $r=R_t(C)$. The established cases in Section~\ref{sec:prelim}, the generalized Singleton bound, and Corollary~\ref{cor:four-pairs} reduce the parameters to
\begin{equation}\label{eq:remaining-range}
\begin{gathered}
8\le\rho\le14,\qquad 3\le t<\rho-4,\qquad
t+1\le r\le\left\lfloor\frac{\rho+t-3}{2}\right\rfloor,\\
3r<2\rho,\qquad
(t,r)\notin\{(3,4),(3,5),(4,5),(4,6)\}.
\end{gathered}
\end{equation}
Theorem~\ref{thm:alphabet} also gives $q<r$. Since $3r<2\rho$ and $\rho\le14$, we have $r\le9$. The only possible prime-power alphabets are therefore
$q\in\{2,3,4,5,7,8\}$.

For each tuple in~\eqref{eq:remaining-range} with $q<r$, compute $K=K(q,\rho,t,r)$ from Lemma~\ref{lem:dimension-cap}. If $K\le5t-2$, the known low-dimension case excludes a counterexample. In every remaining tuple,~\eqref{eq:finite-exclusion} holds. Appendix~\ref{app:finite} specifies the entire finite range and its arithmetic: there are $109$ tuples, of which $27$ satisfy the low-dimension criterion and $82$ satisfy the strict binomial inequality. These calculations were verified with exact integers. Thus every remaining tuple contradicts a necessary condition for a counterexample, proving the theorem.
\end{proof}

\section{A consequence for primitive BCH codes}\label{sec:bch}

The conjecture can also be tested by comparing two nested code families. This gives a short consequence for BCH codes whose redundancy is not bounded by fourteen. Here $\bch(e,m)$ denotes the binary primitive narrow-sense BCH code of length $2^m-1$ and designed distance $2e+1$; $m$ is the extension degree and the alphabet is $\F_2$.

We use the nested-code inequality~\cite[Lemma III.1]{XY2026}:
\begin{equation}\label{eq:nested}
C\subseteq C',\quad \dim C'-\dim C\ge t
\quad\Longrightarrow\quad d_t(C')\le R_t(C).
\end{equation}
To see its direction, choose $t$ vectors of $C'$ independent modulo $C$. Subtracting arbitrary vectors of $C$ leaves $t$ independent vectors of $C'$, whose joint support is at least $d_t(C')$. Their distance to $C$ in the common-support metric gives~\eqref{eq:nested}.

Stable covering estimates for primitive BCH codes were developed in~\cite{BZ2026} and sharpened in~\cite{XYZ2026}. For fixed $e,t\ge2$ and all sufficiently large $m$, the latter gives~\cite[Theorem 1.2]{XYZ2026}
\begin{equation}\label{eq:bch-bounds}
et\le R_t(\bch(e,m))\le(t+1)e-1.
\end{equation}
We use the upper estimate at error parameter $e+1$ and the lower estimate at $e$.

\begin{proposition}\label{prop:bch}
Fix integers $e,t\ge2$. For all sufficiently large extension degrees $m$,
\[
d_t(\bch(e,m))\le(e+1)t+e\le2R_t(\bch(e,m)).
\]
Consequently,
\[
\delta_t(\bch(e,m))\le R_t(\bch(e,m))-1.
\]
\end{proposition}

\begin{proof}
For sufficiently large $m$, the standard binary cyclotomic-coset calculation gives redundancies $em$ and $(e+1)m$ for $\bch(e,m)$ and $\bch(e+1,m)$, respectively~\cite[Sections 1--2]{XYZ2026}. We may also require $m\ge t$. Since $\bch(e+1,m)\subseteq\bch(e,m)$, their dimension difference is $m$. Apply~\eqref{eq:nested} and~\eqref{eq:bch-bounds}, choosing $m$ large enough for both adjacent families. We obtain
\begin{align*}
d_t(\bch(e,m))
&\le R_t(\bch(e+1,m))\\
&\le(t+1)(e+1)-1\\
&=(e+1)t+e\\
&\le2et\\
&\le2R_t(\bch(e,m)).
\end{align*}
The penultimate inequality follows from $(e-1)t-e=(e-1)(t-2)+(e-2)\ge0$. Taking the integer part of $(d_t-1)/2$ gives the asserted strict gap between the radii.
\end{proof}

The comparison $d_t(\bch(e,m))\le2R_t(\bch(e,m))$ also follows from the earlier bounds in~\cite[Theorem 1.1]{BZ2026}. For $e=2$, use its upper bound without the logarithmic term at error parameter three. For $e\ge3$, its general estimate gives, for all sufficiently large $m$,
\[
R_t(\bch(e+1,m))\le(e+1)t+e+\lceil\log_2 t\rceil.
\]
Here $\lceil\log_2 t\rceil\le t-1\le(e-1)t-e$, so the logarithmic term is absorbed by the same comparison with $2et$. The upper estimate in~\eqref{eq:bch-bounds} additionally gives the explicit bound $d_t(\bch(e,m))\le(e+1)t+e$ throughout the stated range.

This consequence fixes $e$ and $t$ before increasing $m$. It uses established bounds for adjacent BCH families, while Theorem~\ref{thm:redundancy} applies to all orders of every code in its redundancy range. Exact second-order covering results for the double- and triple-error-correcting BCH families~\cite{YSch2025,EZ2026} concern a different task: determining the radius itself, rather than comparing it with a generalized Hamming weight.

\section{Conclusion}\label{sec:conclusion}

An auxiliary code with sufficient minimum distance supplies a syndrome space whose short cover forces a subcode with small support. This gives the large-alphabet inequality and, together with classical dual-weight bounds and the established cases of Yu and Schwartz, proves the packing--covering conjecture through redundancy fourteen over every finite field.

Beyond this range, the auxiliary criterion and the dimension cap remain available. The first binary parameter triple not excluded by the reductions used here is $\rho=15$, $t=3$, and $R_3(C)=6$; the cap gives $k\le78$. This identifies a remaining parameter range rather than a counterexample. Possible directions include replacing the auxiliaries that do not exist or obtaining additional restrictions on the short covers of syndrome spaces.

\appendix
\section{The finite inequalities}\label{app:finite}

This appendix specifies the complete set of integer calculations used for
redundancy at most fourteen.  Put
\[
 \mathcal Q=\{2,3,4,5,7,8\},\qquad
 \mathcal E=\{(3,4),(3,5),(4,5),(4,6)\}.
\]
For each $q\in\mathcal Q$, the parameter triples left after the reductions
in the proof are precisely
\begin{equation}\label{eq:finite-universe}
\begin{split}
\mathcal S_q=\{(\rho,t,r)\in\mathbb Z^3:\;&
8\le \rho\le14,\quad 3\le t\le\rho-5,\quad q<r,\\
&t+1\le r\le\lfloor(\rho+t-3)/2\rfloor,\quad
3r<2\rho,\quad (t,r)\notin\mathcal E\}.
\end{split}
\end{equation}
Thus there is no additional length parameter to enumerate.  For a triple
in $\mathcal S_q$, define $a=\rho-2r+t-2$, $b=\rho-2r-1$, and
\begin{equation}\label{eq:finite-cap}
 K=\min_{a<h\le\rho}
 \left\lfloor
 \frac{(q^h-q^{h-a})h-b(q^h-1)}{q^{h-a}-1}
 \right\rfloor,\qquad N=\rho+K.
\end{equation}
The minimum is over integers $h$; whenever a value of $h$ is recorded
below, it is the smallest minimizer.  Separate the triples into
$\mathcal L_q=\{(\rho,t,r)\in\mathcal S_q:K\le5t-2\}$ and
$\mathcal B_q=\mathcal S_q\setminus\mathcal L_q$.
The former are excluded by the low-dimension result.  For the latter,
the required strict inequality is
\begin{equation}\label{eq:finite-ratio}
f_q(\rho,t,r):=\frac{\binom{\rho+K}{r}}{q^{t(\rho-r)}}<1.
\end{equation}

Table~\ref{tab:finite-summary} gives the cardinalities of all three sets
and the exact largest ratio
\[
M_q=\max_{(\rho,t,r)\in\mathcal B_q}f_q(\rho,t,r).
\]
Its final column displays an exact expression for $M_q$, followed by a
simple strict rational bound.  Every one of the $109$ triples belongs
to one of the two exclusion classes: $27$ lie in $\mathcal L_q$ and
$82$ lie in $\mathcal B_q$.
For $q=2,3,4$, the nine triples in $\mathcal L_q$ are exactly those
with $a=1$; choosing $h=3$ already gives
\[
 K\le q(t+1)-3+\left\lfloor\frac{t+1}{q+1}\right\rfloor
 \le5t-2.
\]
For $q=5,7,8$, the set $\mathcal L_q$ is empty.

\begin{table}[htbp]
\centering
\normalsize
\renewcommand{\arraystretch}{1.5}
\setlength{\tabcolsep}{4pt}
\begin{tabular}{rrrrrrrr}
\hline
$q$ & $|\mathcal S_q|$ & $|\mathcal L_q|$ & $|\mathcal B_q|$
& $(\rho,t,r)$ & $h$ & $K$ & $M_q$ and a strict bound\\
\hline
2 & 25 & 9 & 16 & (14,3,6) & 5 & 36 & $\binom{50}{6}/2^{24}<1$ \\
3 & 25 & 9 & 16 & (14,3,6) & 4 & 116 & $\binom{130}{6}/3^{24}<1/40$ \\
4 & 25 & 9 & 16 & (14,3,6) & 4 & 251 & $\binom{265}{6}/4^{24}<1/600$ \\
5 & 25 & 0 & 25 & (14,3,6) & 4 & 464 & $\binom{478}{6}/5^{24}<1/3000$ \\
7 & 7 & 0 & 7 & (14,7,8) & 4 & 2796 & $\binom{2810}{8}/7^{42}<10^{-12}$ \\
8 & 2 & 0 & 2 & (14,8,9) & 3 & 581 & $\binom{595}{9}/8^{40}<10^{-16}$ \\
\hline
\end{tabular}
\medskip
\caption{Complete finite coverage and the largest binomial ratio for
each alphabet.  The recorded triple attains $M_q$, and every displayed
inequality certifies $M_q<1$.}
\label{tab:finite-summary}
\end{table}

For completeness, Table~\ref{tab:finite-binary} lists all sixteen
binary binomial checks individually.  Each entry in its penultimate
column is strictly smaller than $2^s$, where $s=t(\rho-r)$.
The closest comparison is
\[
 \binom{50}{6}=15890700<16777216=2^{24};
\]
the largest nonbinary ratio is attained when $q=3$, and its comparison is
\[
 \binom{130}{6}=5963412000<282429536481=3^{24}.
\]
All entries in both tables can be reproduced from
\eqref{eq:finite-universe}--\eqref{eq:finite-ratio} using integer
arithmetic alone: enumerate the displayed triples, minimize over the
displayed finite interval of $h$, discard $K\le5t-2$, and compare the
remaining ratios by cross multiplication.  In particular,
$\binom Nr=\prod_{j=0}^{r-1}(N-j)/r!$ gives the only binomial
coefficient needed in each comparison.  These are parameter
inequalities, not an enumeration of codes.

\begin{table}[htbp]
\centering
\small
\renewcommand{\arraystretch}{1.15}
\begin{tabular}{rrrrrrrr}
\hline
$\rho$ & $t$ & $r$ & $h$ & $K$ & $N$ & $\binom Nr$ & $s$\\
\hline
11 & 5 & 6 & 4 & 26 & 37 & 2324784 & 25 \\
12 & 5 & 6 & 5 & 57 & 69 & 119877472 & 30 \\
12 & 6 & 7 & 5 & 30 & 42 & 26978328 & 30 \\
13 & 3 & 6 & 4 & 16 & 29 & 475020 & 21 \\
13 & 5 & 6 & 6 & 120 & 133 & 6856577728 & 35 \\
13 & 5 & 7 & 4 & 26 & 39 & 15380937 & 30 \\
13 & 6 & 7 & 6 & 66 & 79 & 2898753715 & 36 \\
13 & 7 & 8 & 5 & 34 & 47 & 314457495 & 35 \\
14 & 3 & 6 & 5 & 36 & 50 & 15890700 & 24 \\
14 & 4 & 7 & 4 & 21 & 35 & 6724520 & 28 \\
14 & 5 & 6 & 7 & 247 & 261 & 414356272512 & 40 \\
14 & 5 & 7 & 5 & 57 & 71 & 1329890705 & 35 \\
14 & 6 & 7 & 7 & 138 & 152 & 323295330680 & 42 \\
14 & 6 & 8 & 5 & 30 & 44 & 177232627 & 36 \\
14 & 7 & 8 & 6 & 75 & 89 & 70625252863 & 42 \\
14 & 8 & 9 & 5 & 39 & 53 & 4431613550 & 40 \\
\hline
\end{tabular}
\medskip
\caption{All binary triples not excluded by $K\le5t-2$.
Every row satisfies $\binom Nr<2^s$.}
\label{tab:finite-binary}
\end{table}

\FloatBarrier
\clearpage
\bibliographystyle{amsplain}
\bibliography{references}
\end{document}